\documentclass[letterpaper, 10 pt, conference]{ieeeconf}  
\usepackage[colorlinks=true, allcolors=blue]{hyperref}
\usepackage{amsfonts}
\usepackage{amsmath}
\usepackage{graphicx}
\usepackage{algorithm}
\usepackage{algpseudocode}
\usepackage{placeins}
\usepackage{tikz}
\usepackage{tabularx}
\usepackage{array}
\usepackage[noadjust]{cite}
\usepackage{subcaption} 

\newcommand{\mathunderbar}[1]{\ensuremath{\underline{\smash{#1}}}}

\IEEEoverridecommandlockouts                              
\title{\LARGE \bf
A Scenario-Based Approach for Stochastic Economic Model Predictive Control with an Expected Shortfall Constraint
}

\author{Alireza Arastou, Algo Car\`e, Ye Wang, Marco Campi, and Erik Weyer 
\thanks{Algo Carè and Marco Campi were partly supported by the PRIN 2022 project 2022RRNAEX ``The Scenario Approach for Control and Non-Convex Design'' (CUP: D53D23001440006), funded by the NextGeneration EU program (Mission 4, Component 2, Investment 1.1), and by the PRIN PNRR project  P2022NB77E “A data-driven cooperative framework for the management of distributed energy and water resources” (CUP: D53D23016100001), funded by the NextGeneration EU program (Mission 4, Component 2, Investment 1.1).} 
\thanks{}
\thanks{Alireza Arastou, and Erik Weyer are with the
Department of Electrical and Electronic Engineering, The
University of Melbourne, Parkville, Melbourne, Victoria, 3010,
Australia, (E-mails: aarastou@student.unimelb.edu.au, 
ewey@unimelb.edu.au)
}%
\thanks{Algo Carè and Marco Campi are with the Department of Information Engineering, University of Brescia, 25123 Brescia,
Italy (E-mails: algo.care@unibs.it; marco.campi@unibs.it).}
\thanks{Ye Wang is with the School of Mathematics and Statistics, The University of Melbourne, Parkville, Melbourne, Victoria, 3010,
Australia, (E-mail: ye.wang@unimelb.edu.au)}
}

\newtheorem{lemma}{\textbf{Lemma}}
\newtheorem{assumption}{\textbf{Assumption}}
\newtheorem{theorem}{\textbf{Theorem}}
\newtheorem{corollary}{\textbf{Corollary}}

\begin{document}

\maketitle
\thispagestyle{empty}
\pagestyle{empty}

\begin{abstract}

This paper presents a novel approach to stochastic economic model predictive control (SEMPC) that minimizes average economic cost while satisfying an empirical expected shortfall (EES) constraint to manage risk. A new scenario-based problem formulation ensuring controlled risk with high confidence while minimizing the average cost is introduced. The probabilistic guarantees is dependent on the number of support elements over the entire input domain, which is difficult to find for high-dimensional systems. A heuristic algorithm is proposed to find the number of support elements. Finally, an efficient method is presented to reduce the computational complexity of the SEMPC problem with an EES constraint. The approach is validated on a water distribution network, showing its effectiveness in balancing performance and risk.

\end{abstract}

\section{INTRODUCTION} \label{section I}

Model predictive control (MPC) is a powerful control approach for a wide range of systems with constraints. The closed-loop properties of deterministic MPC are well-established \cite{rawlings2017model} making it a natural choice for many applications, such as water distribution networks \cite{wang2021minimizing} and traffic systems \cite{kim2014mpc}. In these applications, the control objective can be related to the economic cost incurred during the plant operation. This case is called economic MPC (EMPC) \cite{AngeliEMPC}.

In practice, knowledge about the system is limited or some parameters are uncertain, introducing randomness in the control problem. Thus, the results from a deterministic MPC might not be reliable \cite{HEIRUNG2018158}. One method to cope with bounded uncertainties is to use a robust MPC strategy \cite{bemporad2007robust,rawlings2017model}. This method can maintain stability and performance under worst-case scenarios. However, it can lead to conservative solutions. This is undesirable in applications where a degree of constraint violation is acceptable or the control objective concerns an economic aspect. 
Stochastic MPC (SMPC) addresses the above issues and uses a probabilistic objective function, e.g., expected value of the control objective. Moreover, hard constraints are replaced with stochastic chance constraints to allow a degree of constraint violation \cite{farina2016stochastic, HEIRUNG2018158}. Uncertainties and disturbances can be unbounded in this case. An SMPC with chance constraints is usually not computationally tractable. Two methods can be used to solve an SMPC problem: analytical approaches and sample-based (scenario-based) methods \cite{farina2016stochastic}. An analytical approach makes use of a priori knowledge about the probability distribution of uncertainties to find a deterministic equivalent of the SMPC problem \cite{HEIRUNG2018158}. A scenario-based method uses samples of the uncertain elements as approximations of their distributions and reformulates the SMPC problem \cite{campi2019scenario}. Bounds are given on the number of samples required such that the solution of a scenario-based MPC meets chance constraints with a desired confidence level for a new scenario \cite{campi2019scenario}.

The focus of this paper is on the scenario-based approach. However, different from the standard scenario theory, we will minimize an average economic cost subject to a constraint on the risk of incurring a very large economic cost. Various risk measures have been developed to quantify potential losses under uncertainty \cite{artzner1999coherent}. Among these, Value-at-Risk (VaR), Conditional Value-at-Risk (CVaR), and Expected Shortfall (ES) are widely used due to their intuitive interpretation and mathematical tractability \cite{rockafellar2000optimization}.

In this paper, a stochastic EMPC (SEMPC) problem with an empirical expected shortfall (EES) constraint is considered. The solution to the SEMPC does not lend itself to be effectively studied by a direct application of the results in \cite{campi2018wait,campi2008exact,campi2019scenario, garatti2022risk}. Moreover, EES is considered as a constraint; thus, the suggested method in \cite{arici2021theory,RAMPONI20181003} cannot be used. The main contributions of the paper are \textit{(i) Probabilistic guarantees on the economic cost for a next unseen scenario, (ii) A heuristic algorithm to find the number of support elements within the feasibility region of the SEMPC problem, (iii) A method to reduce the computational complexity of solving the scenario-based SEMPC.}

The proposed approach is applied to the Richmond water distribution network and the results are discussed. The rest of the paper is organized as follows: the problem statement is given in Section II together with a summary of the scenario approach. Section III gives the probabilistic guarantees. A heuristic algorithm to find the number of support elements within the whole input domain is given in Section IV. An efficient method to reduce computational complexity of solving the obtained control problem is presented in Section V. Simulation results are given in Section VI and conclusions and discussions are given in Section VII.  


\section{PROBLEM STATEMENT}

A stochastic optimization problem is formulated such that an average cost is minimized subject to a constraint on EES. The problem formulation is practically useful in applications such as water distribution networks, where balancing the risk of occasional high costs due to high electricity prices against low average costs is important.  

\subsection{Motivating Example}

Consider an EMPC problem for minimization of pumping energy cost in a water distribution network \cite{wang2021minimizing}:
\begin{subequations}\label{eq: EMPC motivating ex}
\begin{align}
&\hspace{-2 cm}\min_{u_{0|t},\ldots, u_{N-1|t}}  \sum_{k=0}^{N-1} \alpha_{k|t}^{\top}u_{k|t},\label{eq: Electricity cost for the motivating ex}\hspace{0 cm} +\Delta u_{l|t}^{\top}R\Delta u_{l|t}+V^f(x_{N|t}) \nonumber \\
 \hspace{0 cm}\text{subject to}& \hspace{0.3 cm} \forall k \in \mathbb{I}_{[0:N-1]},  \nonumber\\
& x_{0|t}=x_0,\\
& x_{k+1|t}=Ax_{k|t}+B_{u}u_{k|t}+B_{d}d_{k|t},\\
& x_{k|t} \in \mathcal{X}, \hspace{0.5 cm} u_{k|t} \in \mathcal{U}, \hspace{0.3 cm} x_{N|t}\in \mathcal{X}^f, \label{eq: state, input, and terminal sets}
\end{align}
\end{subequations}
where $x_{k|t} \in \mathbb{R}^{n \times 1}$ are states representing water levels in tanks at time $k+t$ with an initial time $t$, $u_{k|t} \in \mathbb{R}^{m\times 1}$ are water flows through pumps, $d_{k|t} \in \mathbb{R}^{v\times 1}$ are water demands, and $\alpha_{k|t} \in \mathbb{R}^{m \times 1}$ is the vector of electricity prices. The objective function includes the energy cost $\alpha_{k|t}^{\top}u_{k|t}$ which is subject to stochastic fluctuations and $\Delta u_{k|t}=u_{k|t}-u_{k-1|t}$ emphasizing input smoothness. $\mathcal{X}^f$ and $V^f(x_{N|t})$ represent terminal constraint and cost. State and input constraints are given by $\mathcal{X}$ and $\mathcal{U}$. $A\in \mathbb{R}^{n \times n}$, $B_{u} \in \mathbb{R}^{n \times m}$, $B_{d} \in \mathbb{R}^{n \times v}$ are known matrices.  

Denote the cost function in \eqref{eq: Electricity cost for the motivating ex} by $J(\tilde{\alpha}_{t},\tilde{u}_{t})$, where $\tilde{\alpha}_{t}=[\alpha_{0|t}^{\top}, \ldots, \alpha_{N-1|t}^{\top}]^{\top}$ and $\tilde{u}_{t}=[u_{0|t}^{\top}, \ldots, u_{N|t}^{\top}]^{\top}$. In some cases water companies buy electricity directly from an electricity market and $\tilde{\alpha}_t$ and the cost in \eqref{eq: EMPC motivating ex} are stochastic. 
The operational objective is to minimize pumping cost in the long run and the criterion
\begin{equation}
    \min_{\tilde{u}_t}  \mathbb{E}\left\{ J(\tilde{\alpha}_{t},\tilde{u}_{t})\right\} \label{eq: expected motivating ex}
\end{equation}
is used. A shortcoming of \eqref{eq: expected motivating ex} is that very high energy prices which can happen with a small probability may occasionally lead to unacceptably high operating cost.  

\subsection{Risk Measures}\label{sec:2.B}

Risk measures, such as value at risk (VaR) and expected shortfall (ES), are statistical tools used to evaluate risk \cite{artzner1999coherent}. Let $L(u,\delta)$ be a random variable where $u$ is the decision variable and $\delta$ is the uncertainty. Let $F_{L,u}$ be cumulative distribution function (CDF) of $L(u,\cdot)$. Given a $ \zeta \in (0,1) $, VaR and ES at level $\zeta$ are given by
\begin{subequations}
    \begin{align}
        & \text{VaR}_{\zeta}(L_u) = \min \{ l : F_{L,u}(l) \geq \zeta \} \label{eq: VaR definition},\\
        & \text{ES}_{\zeta}(L_u) = \mathbb{E} \{ L_u : L_u \geq \text{VaR}_{\zeta}(L_u) \} \label{eq: ES definition}.
    \end{align}
\end{subequations} 

The distribution function $F_{L,u}$ is required to compute the risk measures in \eqref{eq: VaR definition}-\eqref{eq: ES definition}; however, it is not known in many applications. Therefore, the empirical versions of \eqref{eq: VaR definition}-\eqref{eq: ES definition} using samples (scenarios) of the random variable is used. In this paper, we focus on EES as suggested in \cite{arici2021theory,RAMPONI20181003}. Given $N_s$ independent realizations of the random variable $(\delta^1,\ldots,\delta^{N_s})$, we define $L_i(u):=L(u,\delta^i)$. For a given $u \in \mathcal{U}$, $L_{(i)}(u), i=1,\ldots,N_s$ are the loss functions, $L_i$, in descending order \cite{RAMPONI20181003}
\begin{equation} \label{eq: loss functions in descending order}
    L_{(1)}(u)\geq L_{(2)}(u) \geq \ldots \geq L_{(N_s)}(u),
\end{equation}

EES at level $1-\frac{k}{N_s}$ is the average of $k$-largest realizations,
\begin{equation} \label{eq: EES definition}
    \text{EES}_{1-\frac{k}{N_s}}(L(u))=\frac{1}{k} \sum_{i=1}^k L_{(i)}(u).
\end{equation}

It is desirable to limit the EES to control the effect of worst case scenarios. SEMPC in \eqref{eq: EMPC motivating ex}-\eqref{eq: expected motivating ex} is hence combined with EES in \eqref{eq: EES definition} as a constraint in a scenario optimization problem. Let $\tilde{\alpha}^i_{t}$ be $N_s$ independently drawn scenarios of the vector of electricity prices over the prediction horizon $N$. The following problem is considered:
\begin{subequations}\label{eq: SEMPC with EES}
\begin{align}
&\min_{u_{0|t},\ldots, u_{N-1|t}} \frac{1}{N_s}\sum_{i=1}^{N_s} (\sum_{l=0}^{N-1} (\alpha_{l|t}^i)^\top u_{l|t}  \label{eq: expectd cost in SEMPC} \\
&\hspace{2 cm}+\Delta u_{l|t}^{\top}R\Delta u_{l|t}+V^f(x_{N|t})) \nonumber\\
 \hspace{0 cm}\text{subject to}& \hspace{0.3 cm} \forall l \in \mathbb{I}_{[0:N-1]},  \nonumber\\
& x_{0|t}=x_0,\\
& x_{l+1|t}=Ax_{l|t}+B_{u}u_{l|t}+B_{d}d_{l|t},\\
& x_{l|t} \in \mathcal{X}, \hspace{0.5 cm} u_{l|t} \in \mathcal{U},\hspace{0.2 cm} x_{N|t}\in \mathcal{X}^f, \label{eq: terminal set}\\
& \frac{1}{k} \sum_{j=0}^k \left(\sum_{l=0}^{N-1} (\alpha_{l|t}^{i_j})^\top u_{l|t}\right)\leq M, \hspace{0.1 cm} \label{eq: EES constraint in SEMPC}\\
& \text{for any choice of} \hspace{0.1 cm} \{i_1,\ldots,i_k\} \subseteq \{1,\ldots,N_s\} \nonumber,
\end{align}
\end{subequations}
where \eqref{eq: EES constraint in SEMPC} specifies an upper bound of $M$ on EES. $\{i_1,\ldots,i_k\}$ in \eqref{eq: EES constraint in SEMPC} is any subset of $\{1,\ldots,N_s\}$ with cardinality $k$; however, only the constraint \eqref{eq: EES constraint in SEMPC} with indices corresponding to the $k$-largest scenarios is active. 

A natural question is whether the obtained solution is reliable given an unseen scenario. Fundamentals of the scenario approach and the challenges of using this method for the problem in \eqref{eq: SEMPC with EES} are given in the next section.

\subsection{Scenario Theory for General Decision Making \cite{garatti2022risk}} \label{sec: decision-making scenario}
Let $\Phi_m$ be a map from a set of scenarios $(\delta^1,\ldots,\delta^{m})$ to a decision $z \in \mathcal{Z}$, where $\mathcal{Z}$ is a generic decision set. For every $\delta$ there is a set $\mathcal{Z}_{\delta}$ of allowed decisions. Assume that $\Phi_m$ has the following properties \cite{campi2021scenarioReview,garatti2022risk}:

\noindent \begin{assumption} \label{As: Map properties}
(Mapping properties)

    \begin{itemize}
        \item \textbf{Permutation invariance}: For a permutation of $(\delta^1,\ldots,\delta^m)$ denoted by $(\delta^{i_1},\ldots,\delta^{i_m})$, we have $\Phi_m(\delta^1,\ldots,\delta^m)=\Phi_m(\delta^{i_1},\ldots,\delta^{i_m})$.
        \item \textbf{Stability in the case of confirmation}: For any integer $n$, if the set of scenarios given by $(\delta^1,\ldots,\delta^m,\delta^{m+1},\ldots,\delta^{m+n})$ is such that 
        \begin{equation}
            \Phi_m(\delta^1,\ldots,\delta^m) \in \mathcal{Z}_{\delta^{m+i}}, \hspace{0.3 cm} \forall i \in\{1,\ldots,n\}, \nonumber
        \end{equation}
        then $\Phi_{m+n}(\delta^1,\ldots,\delta^{m+n})=\Phi_m(\delta^1,\ldots,\delta^m)$.

    \item \textbf{Responsiveness to contradiction}: Let $(\delta^1,\ldots,\delta^m,\delta^{m+1},\ldots,\delta^{m+n})$ be a set of scenarios where $n$ is an integer, such that 
        \begin{equation}
            \exists i \in\{1,\ldots,n\}: \Phi_m(\delta^1,\ldots,\delta^m) \notin \mathcal{Z}_{\delta^{m+i}},  \nonumber
        \end{equation}
        then $\Phi_{m+n}(\delta^1,\ldots,\delta^{m+n})\neq \Phi_m(\delta^1,\ldots,\delta^m)$.
    \end{itemize}
\end{assumption}

Risk is defined as $V(z)=\mathbb{P}\{\delta \in \Delta: z \notin \mathcal{Z}_{\delta}\}$. Also, $\mathcal{Z}_{\delta^i}$ is called a support element if $\Phi_{m-1}(\delta^1,\ldots,\delta^{i-1},\delta^{i+1},\ldots,\delta^m) \neq \Phi_{m}(\delta^1,\ldots,\delta^{m})$ \cite{garatti2022risk}. The following non-degeneracy assumption is in place. 
\vspace{0.1 cm}
\begin{assumption} \label{As: Non-degeneracy assumption}
    Almost surely $\Phi_m(\delta^1,\ldots,\delta^m)$ coincides with the obtained decision after eliminating all elements that are not support elements. 
\end{assumption}

The following theorem provides a probabilistic certificate for the risk at $z^*_m=\Phi_m(\delta^1,\ldots,\delta^m)$.

\vspace{0.2 cm}

\begin{theorem}[\cite{garatti2022risk}] \label{Theorem: Risk guarantee theorem}
Assume that Assumptions \ref{As: Map properties} and \ref{As: Non-degeneracy assumption} hold. Let $\beta \in (0,1)$ be a confidence parameter. For each $k = 0, 1, \dots, m-1$, consider the polynomial equation
\begin{equation} \label{eq: binom equation 1}
\binom{m}{k} t^{m-k} - \frac{\beta}{2m} \sum_{i=k}^{m-1} \binom{i}{k} t^{i-k} - \frac{\beta}{6m} \sum_{i=m+1}^{4m} \binom{i}{k} t^{i-k} = 0,
\end{equation}
which has exactly two solutions in $[0, +\infty)$, denoted by $\mathunderbar{t}(k)$ and $\overline{t}(k)$, with $\mathunderbar{t}(k) \leq \overline{t}(k)$.  For $k = m$, consider the polynomial equation
\begin{equation} \label{eq: binom equation 2}
1 - \frac{\beta}{6m} \sum_{i=m+1}^{4m} \binom{i}{k} t^{i-m} = 0,
\end{equation}
which has a unique solution $\overline{t}(m)$. Let $\mathunderbar{t}(m) = 0$. Set
$\mathunderbar{\epsilon}(k) := \max\{0, 1 - \overline{t}(k)\},$ and $ \overline{\epsilon}(k) := 1 - \mathunderbar{t}(k), \hspace{0.1 cm} k = 0, 1, \dots, m$. Let $s^*_m$ be the number of support elements at $z_m^*$. Under Assumptions \ref{As: Map properties} and \ref{As: Non-degeneracy assumption}, it holds that
\begin{equation} \label{eq: lower and upperbounds for risk}
    \mathbb{P}^m \left\{ \mathunderbar{\epsilon}(s^*_m) \leq V(z^*_m) \leq \overline{\epsilon}(s^*_m) \right\} \geq 1 - \beta.
\end{equation}
\end{theorem}

\vspace{0.2 cm}

\begin{proof}
The proof is given in \cite{garatti2022risk}.    
\end{proof}

In our case we are interested in a bound on the energy cost $\sum_{l=0}^{N-1}\alpha_{l|t}^{\top}u_{l|t}$ for a new unseen scenario of the electricity prices (the real cost that will be incurred). In particularly, we would like this cost to be less than $M$. Bounds of this type was given in \cite{RAMPONI20181003} for the solution $u^*_{l|t}$ which minimized the EES on the seen $N_s$ scenarios. However, this theory is not applicable to the situation at hand since the solution here is obtained by minimizing a different objective function and the EES is acting as a constraint only.
\section{Probabilistic Guarantees} \label{sec: guarantees for the whole region}

To circumvent the above problem, probabilistic certificates for the EES for the whole input domain are found. If the cost associated with a new scenario is not among the largest $k$ costs for any $u$ in the input domain, EES will not change for any input including the solution of \eqref{eq: SEMPC with EES}. Therefore, \eqref{eq: EES constraint in SEMPC} is met given the added scenario. Hence, we are interested in finding guarantees for the situation that the cost associated with a new scenario is not among the largest $k$ costs for any input and hence EES is unchanged. The probability of violation for a new scenario is the probability that the new scenario results in a cost larger than the $k$-largest cost for at least one input. In the following, we will find a certificate for this probability of violation.  

Denote the region below $k$-largest cost by $D$, i.e.,
\begin{equation} \label{eq: feasibility problem}
    D=\{(\tilde{u},l) \in \mathcal{U}^N \times \mathbb{R} \hspace{0.0 cm}:\hspace{0.1 cm}\tilde{u} \in \mathcal{U}^N, \hspace{0.1 cm} 0\leq l \leq L_{(k)}(\tilde{u})\}.
\end{equation}

An example for a scalar $u$ with $\mathcal{U}=\{u \hspace{0.1 cm} : \hspace{0.1 cm} \mathunderbar{u} \leq u\leq \Bar{u}\}$, $N_s=4$, and $k=2$ is shown in Fig. \ref{fig: example of feasibility problem}.

\begin{figure}[t]
    \centering
    \includegraphics[scale=0.3]{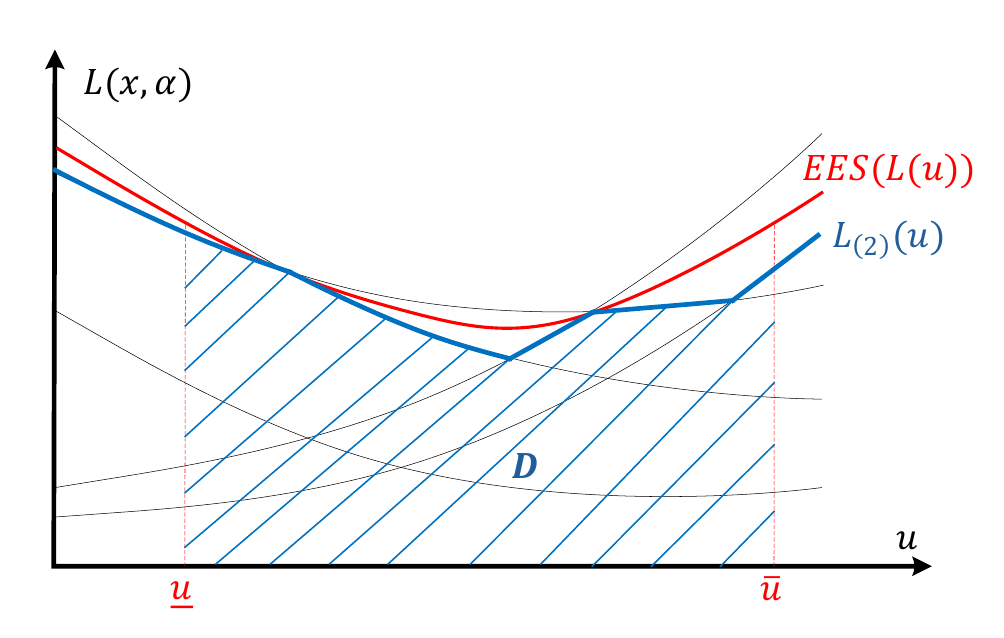}
    \caption{Region $D$ in \eqref{eq: feasibility problem} with $N_s=4$, $k=2$. The blue line is the second-largest cost, while gray lines are $L_i(u)$, $i=1,2,3,4$.}
    \label{fig: example of feasibility problem}
\end{figure}

If adding a new scenario to \eqref{eq: feasibility problem} does not change the set $D$, then the added scenario was not among the largest $k$ costs for any $\tilde{u} \in \mathcal{U}^N$. Therefore, $\text{EES}_{1-\frac{k}{N_s}}(L(\tilde{u}))$ is unchanged for the whole input domain. The probability of violation for the set $D$ is defined as 
\begin{equation} \label{eq: probability of violation in words}
    \mathbb{P}\{\tilde{\alpha} \in \Delta \hspace{0.0 cm}:\hspace{0.0 cm}\text{$\exists \tilde{u} \in \mathcal{U}^N$ such that $(\tilde{u},L(\tilde{u},\tilde{\alpha})) \notin D$}\}
\end{equation}

$D=\Phi_m(\tilde{\alpha}^1,\ldots,\tilde{\alpha}^m)$ can be viewed as a decision in the framework of Section \ref{sec: decision-making scenario} and the corresponding theory can be applied to obtain the probabilistic guarantees. Formally, in the context of EES, we can define the sets
\begin{subequations} \label{eq: definitions of mapping sets}
    \begin{align}
        & \mathcal{Z}=\{\bar{\mathcal{Z}} \hspace{0.0 cm}:\hspace{0.0 cm} \bar{\mathcal{Z}} \subseteq \mathcal{U}^N\times\mathbb{R}^+\},\\
        & {\cal Z}_{\tilde{\alpha}}:=\{  {\cal S} \in {\cal Z} : \forall \tilde{u}\in \Pi_{\mathcal{U}^N}(\mathcal{S}), \hspace{0.1cm}(\tilde{u}, L(\tilde{u},\tilde{\alpha}))\in {\cal S}\}
    \end{align}
\end{subequations}
where $\Pi_{\mathcal{U}^N}(\mathcal{S})$ is the projection of $\mathcal{S}$ on $\mathcal{U}^N$. 
For the decision $D$, the probability of violation is 
\begin{equation} \label{eq: probability of violation feasibility problem}
    V(D)=\mathbb{P}\{\tilde{\alpha} \in \Delta: D \notin \mathcal{Z}_{\tilde{\alpha}}\}.
\end{equation}


If the corresponding cost of a scenario is above $D$ for at least one $\tilde{u} \in \mathcal{U}^N$, then $D \notin \mathcal{Z}_{\tilde{\alpha}}$. The following theorem provides guarantees for $V(D)$ and consequently for the solution of \eqref{eq: SEMPC with EES}.



    
    

\vspace{0.1 cm}

\begin{theorem} \label{theorem: main theorem}
    Assume Assumption \ref{As: Non-degeneracy assumption} holds. Given a set of scenarios $(\tilde{\alpha}^1,\ldots,\tilde{\alpha}^{N_s})$, let the solution of \eqref{eq: feasibility problem} be given by $D$ and denote the number of support elements within the whole input domain by $s^*_{N_s}$, then, it holds that
    \begin{equation} \label{eq: lower and upperbounds for risk of feasibility problem}
    \mathbb{P}^{N_s} \left\{ \mathunderbar{\epsilon}(s^*_{N_s}) \leq V(D) \leq \overline{\epsilon}(s^*_{N_s}) \right\} \geq 1 - \beta, 
\end{equation}
where $\beta \in (0, 1)$ and $\mathunderbar{\epsilon}(s^*_{N_s})$ and $\overline{\epsilon}(s^*_{N_s})$ are computed from \eqref{eq: binom equation 1}-\eqref{eq: binom equation 2}. 
\end{theorem}

\vspace{0.2 cm}

\begin{proof}
    It can be verified that the map in \eqref{eq: feasibility problem} meets the conditions in Assumption \ref{As: Map properties} . Under Assumption \ref{As: Non-degeneracy assumption}, Theorem \ref{Theorem: Risk guarantee theorem} holds and the certificate in \eqref{eq: lower and upperbounds for risk of feasibility problem} is established. 
\end{proof}

The corollary below shows that \eqref{eq: lower and upperbounds for risk of feasibility problem} can be used to find a certificate for EES at the solution of \eqref{eq: SEMPC with EES}.
\begin{corollary} \label{coro: Guarantee for solution of SEMPC}
Given an input $\tilde{u}
\in \mathcal{U}^N$, let
    $\bar{V}(\tilde{u})=\mathbb{P}\{\tilde{\alpha} \in \Delta \hspace{0.1 cm} : \hspace{0.1 cm} L(\tilde{\alpha},\tilde{u})>\text{EES}_{1-\frac{k}{N_s}}(\tilde{u})\}$. Also, let the solution of \eqref{eq: SEMPC with EES} using a set of scenarios $(\tilde{\alpha}^1,\ldots,\tilde{\alpha}^{N_s})$, be given by $\tilde{u}^*$. It holds that 
    \begin{equation}
        \mathbb{P}^{N_s} \left\{ \bar{V}(\tilde{u}^*) \leq \overline{\epsilon}(s^*_{N_s}) \right\} \geq 1-\beta.
    \end{equation}
\end{corollary}

\vspace{0.2 cm}

\begin{proof} 
We know that 
\begin{equation}\nonumber
\begin{split}
    & \bar{V}(\tilde{u}^*) \leq \mathbb{P}\{\tilde{\alpha} \in \Delta \hspace{0.0 cm}:\hspace{0.0 cm} L(\tilde{\alpha},\tilde{u}^*)>L_{(k)}(\tilde{u}^*)\}\\
    &\leq \mathbb{P}\{\tilde{\alpha} \in \Delta \hspace{0.0 cm}:\hspace{0.0 cm} \exists \tilde{u} \in \mathcal{U}^N \text{s.t.}\hspace{0.1 cm} L(\tilde{\alpha},\tilde{u})>L_{(k)}(\tilde{u})\}=V(D)
\end{split}
\end{equation}
It can be concluded that
\begin{equation} \label{eq: guarantee for u^*}
\begin{split}
    & \mathbb{P}^{N_s} \left\{ \bar{V}(\tilde{u}^*) \leq \overline{\epsilon}(s^*_{N_s}) \right\} \geq \mathbb{P}^{N_s} \left\{ V(D) \leq \overline{\epsilon}(s^*_{N_s}) \right\}\\
    & \geq \mathbb{P}^{N_s} \left\{ \mathunderbar{\epsilon}(s^*_{N_s}) \leq V(D) \leq \overline{\epsilon}(s^*_{N_s}) \right\} \geq 1 - \beta
\end{split}
\end{equation}
\end{proof}

\section{Finding the Number of Support Elements} \label{sec: finding support elements}


 In \eqref{eq: feasibility problem}, the support elements are the scenarios that are among the largest $k$ costs for at least one $\tilde{u} \in \mathcal{U}^N$. Thus, one method to find support elements is to evaluate costs for the whole input domain, order them, and find the largest $k$ costs for every input. However, this method is not computationally feasible for high dimensional systems. In this section, we propose a sample-based method to find the number of support elements for the whole input domain. 

The idea is to draw i.i.d samples from the input domain, find the number of support elements using the obtained samples and test the obtained solution on additional samples of the input. If the obtained solution passes the test, the algorithm will output the obtained number of support elements. Otherwise, new scenarios are added to the set of support elements. The proposed method is given in Algorithm \ref{alg: support elements algorithm}.  

$\hat{p}=(\text{number of support elements}) / N_T$ is the empirical probability of finding a new support element when tested against $N_T$ new input samples. Let $\mu-\rho$ be a desired upper bound on the probability of finding a new support element. Thus, it is desirable to set a very small value for $\mu-\rho \in [0,1]$, so that the probability that we get a new number of support elements for new samples of input be negligible. The break criterion in Algorithm \ref{alg: support elements algorithm} is motivated by the following result, which also provides a guideline for how $N_T$ can be chosen. 
\begin{lemma} \label{lem: number of testing samples}
 Choose a distribution on ${\mathcal{U}}^N$, from which $\tilde{ u}$ is independently sampled. Let $A \subset \{ \tilde{\alpha}_1,\ldots,\tilde{\alpha}_{N_s}\}$ be a set of support elements. Let $p$ be the probability of drawing a $\tilde{u}$ such that there is an $\tilde{\alpha}^* \in \{\tilde{\alpha}_1,\ldots,\tilde{\alpha}_{N_s}\}\setminus A$ with the property that $(\tilde{\alpha}^*)^{\top}\tilde{u}$ is among the $k$-largest values of $(\tilde{\alpha}^i)^{\top}\tilde{u}$, $i=1,\ldots N_s$. Let $\hat{p}$ be the empirical frequency of this event evaluated on an i.i.d. sample $\tilde{u}_1,\ldots,\tilde{u}_{N_T}$. If $N_T$ satisfies

\begin{equation} \label{eq: number of testing samples from input}
      \sum_{i=0}^{\lfloor N_T(\mu - \rho) \rfloor} \binom{N_T}{i} \mu^{i} (1 - \mu)^{N_T - i} < \bar{\beta},
\end{equation}
it holds that $P^{N_T}\{\hat{p}\leq \mu-\rho \hspace{0.1 cm} \text{and}\hspace{0.1 cm} p>\mu\}<\bar{\beta}$, where $P^{N_T}$ is the probability measure on the input samples.
\end{lemma}

\begin{proof}
    The proof is similar to that of Lemma 1 in \cite{nasir2018scenario}. 
\end{proof}

The bound in Lemma \ref{lem: number of testing samples} applies to the procedure run at a single iteration of the while loop for one instance of ${\cal I}_k$. However, note that during an execution of Algorithm \ref{alg: support elements algorithm} several versions of ${\cal I}_k$ can be tested and the returned ${\cal I}_k$ depends on the outcomes of a number of tests (an {\em a priori} unknown quantity); therefore, the bound in Lemma \ref{lem: number of testing samples} is not directly applicable to the set of support elements ${\cal I}_k$ returned by the algorithm. 
 
\color{black}

The input domain is restricted by the constraints in \eqref{eq: SEMPC with EES}. Thus, if the input samples are drawn from the feasibility region of SEMPC problem, the number of support elements may be reduced in comparison with the candidate number of support elements obtained by sampling from the whole input domain, i.e., $\mathcal{U}^N$ and we can get tighter guarantees from \eqref{eq: lower and upperbounds for risk of feasibility problem}.

Let us denote the restricted input domain from \eqref{eq: state, input, and terminal sets} by $\mathcal{P}$. Given an initial state $x_0$ and a steady-state $x^s$, assume in \eqref{eq: state, input, and terminal sets} that $\mathcal{X}=\{x\hspace{0.0 cm}:\hspace{0.0 cm}\mathunderbar{x}\leq x \leq \overline{x}\}$, $\mathcal{U}=\{u\hspace{0.0 cm}:\hspace{0.0 cm}\mathunderbar{u}\leq u \leq \overline{u}\}$, and $\mathcal{X}^f=\{x \hspace{0.0 cm}:\hspace{0.0 cm} (x-x^s)^{\top} \Omega (x-x^s)\leq \kappa \}$, where $\Omega$ is a positive definite matrix and $\kappa >0$. $\mathcal{P}$ can be represented by
\begin{equation} \label{eq: feasible set}
   \begin{split}
      & \tilde{u} \in \mathcal{U}^N, \hspace{0.1 cm} \bar{B} \tilde{u} \leq \bar{A},\hspace{0.1 cm} (\hat{B} \tilde{u} + \gamma)^{\top} \Omega (\hat{B} \tilde{u} + \gamma) \leq \kappa,
   \end{split}
\end{equation}
where $\bar{A}, \bar{B}$, $\hat{B}$, and $\gamma$ are matrices and vectors with appropriate dimensions that include model dynamics, $x_0$, and $x^s$.

An optimization-based method is proposed in this part to find support elements within $\mathcal{P}$. Let us assume $\mathcal{I}_k$ contains the indices of candidate support elements within the box of input constraint, i.e., $\mathcal{U}=\{u\hspace{0.1 cm}:\hspace{0.1 cm}\mathunderbar{u}\leq u \leq \overline{u}\}$. For an index $j \in \mathcal{I}_k$, if there exists a $\tilde{u} \in \mathcal{P}$, such that $L(\tilde{u},\tilde{\alpha}^j)$ is among the largest $k$ values, $j$ is a support element in $\mathcal{P}$. Otherwise, we can remove $j$ from $\mathcal{I}_k$. The following program is solved for every $j \in \mathcal{I}_k$
\begin{subequations} \label{eq: Finding support in P}
\begin{align}
    &\max_{\tilde{u},\{z_i\}} \;  \sum_{i=1,i\neq j}^{N_s} z_i, \label{eq: objective function z_i} \\
    \text{subject to:} \hspace{0.2 cm} &i=1,\ldots,N_s, \hspace{0.1 cm} i\neq j, \nonumber\\
    & \tilde{u} \in \mathcal{P},\hspace{0.2 cm} (\tilde{\alpha}^j)^{\top} \tilde{u} \geq (\tilde{\alpha}^i)^{\top} \tilde{u} - G (1 - z_i), \label{eq: verifying that j is above i}\\
    & \sum_{i=1}^{N_s} z_i \geq N - k, \label{eq: sum of z_i},\hspace{0.2 cm} z_i \in \{0, 1\},
\end{align}
\end{subequations}
\begin{algorithm}[h]
\caption{Finding the number of support elements}\label{alg: support elements algorithm}
\begin{algorithmic}[1]
\State Inputs: $(\tilde{\alpha}^1,\ldots,\tilde{\alpha}^{N_s})$, $N_r$, $N_T$, $k$, $\mu$, $\rho$;
\State Generate $N_r$ i.i.d samples according to a uniform distribution on $\mathcal{U}^N$, i.e., $(\tilde{u}_1,\ldots,\tilde{u}_{N_r})$;
\State Evaluate $L(\tilde{u}_j,\tilde{\alpha}^i)$ for all $i=1,\ldots,N_s$ and $j=1,\ldots,N_r$;
\State Find the indices of the scenarios that are among the largest $k$ costs and store them in a set $\mathcal{I}_k$;  
\While{$|\mathcal{I}_k|\leq N_s$}
\State Generate $N_T$ new input samples according to a uniform distribution on $\mathcal{U}^N$ for testing $\mathcal{I}_k$, i.e., $(\tilde{u}_1,\ldots,\tilde{u}_{N_T})$;
\For{each $\tilde{u}_l$, $l=1,\ldots,N_T$}
\State Find the $k$ scenarios with largest values and store them in $\bar{\mathcal{I}}_l$;
\State Define $B_{l} = \begin{cases}
0, & \bar{\mathcal{I}}_l \subset \mathcal{I}_k, \\
1, & \text{otherwise.}
\end{cases}$
\State If $B_l = 1$, add the new support elements to $\mathcal{I}_k$;
\EndFor
\State Compute $\hat{p}=\frac{1}{N_T}\sum_{l=1}^{N_T} B_l$;
\If {$\hat{p}\leq \mu-\rho$}
\State break;
\EndIf
\EndWhile
\State Output: A set of support elements $\mathcal{I}_k$. 
\end{algorithmic}
\end{algorithm}where $G$ is a positive large number. If the above problem is feasible for a $j \in \mathcal{I}_k$, $j$ is a support element in $\mathcal{P}$.  
If $(\tilde{\alpha}^j)^{\top} \tilde{u} \geq (\tilde{\alpha}^i)^{\top} \tilde{u}$, \eqref{eq: verifying that j is above i} will be met with either $z_i=1$ or $0$. Since the objective function in \eqref{eq: objective function z_i} is maximized, then $z_i=1$ wherever $(\tilde{\alpha}^j)^{\top} \tilde{u} \geq (\tilde{\alpha}^i)^{\top} \tilde{u}$. If $(\tilde{\alpha}^j)^{\top} \tilde{u} < (\tilde{\alpha}^i)^{\top} \tilde{u}$, $z_i=0$ in \eqref{eq: verifying that j is above i}. Therefore, $\sum_{i=1,i\neq j}^{N_s} z_i$ is the number of scenarios that are below $j$-th scenario and it is required to be greater than $N-k$ in \eqref{eq: sum of z_i} to make sure the $j$-th scenario is among top $k$ scenarios. Feasibility of \eqref{eq: Finding support in P} implies that the scenario index by $j$ is of support. 

\section{Computationally Tractable SEMPC Problem}

The obtained problem in \eqref{eq: SEMPC with EES} is a convex quadratic program. However, considering every possible choice of $k$ indices from ${1,\ldots,N_s}$ makes the computational burden high. An equivalent computationally suitable formulation can be obtained by using \cite{OGRYCZAK2003117}
\begin{lemma} \label{lem: lemma for sum of k-largest costs}
    Given $\tilde{u} \in \mathcal{U}^N$ and a collection of scenarios $\{L_i(\tilde{u})\}_{i=1}^{N_s}$, the sum of $k$-largest functions is given by
    \begin{subequations} \label{eq: sum of k-largest functions equivalence}
        \begin{align}
            &\min_{\bar{t},\{\lambda_i\}_{i=1}^{N_s}} k\bar{t}+\sum_{i=1}^{N_s} \lambda_i\\
            \text{subject to}:& \hspace{0.1 cm} \lambda_i\geq L_i(\tilde{u})-\bar{t}, \hspace{0.2 cm} \lambda_i \geq 0, \hspace{0.1 cm} i=1,\ldots,N_s.
        \end{align}
    \end{subequations}
\end{lemma}

\vspace{0.1 cm}

\begin{proof}
    See \cite{OGRYCZAK2003117}.
\end{proof}

In the SEMPC problem in \eqref{eq: SEMPC with EES}, \eqref{eq: EES constraint in SEMPC} can be replaced by
\begin{subequations}\label{eq: Equivalent of SEMPC with EES}
\begin{align}
& \lambda_j \geq \sum_{l=0}^{N-1} (\alpha_{l|t}^{j})^\top u_{l|t}-\bar{t}, \hspace{0.1 cm} \lambda_j\geq 0, \hspace{0.1 cm}, j =1,\ldots,N_s \label{eq: lambda for sum k inequality}\\
& \frac{1}{k}\left(k\bar{t}+\sum_{j=1}^{N_s} \lambda_j\right)\leq M \label{eq: equiavlent EES <M}.
\end{align}
\end{subequations}

\eqref{eq: lambda for sum k inequality}-\eqref{eq: equiavlent EES <M} reduces the computational burden compared to \eqref{eq: EES constraint in SEMPC}. $t$ and $\{\lambda_i\}_{i=1}^{N_s}$, which are obtained from \eqref{eq: Equivalent of SEMPC with EES} are not necessarily the minimizer of \eqref{eq: sum of k-largest functions equivalence}; thus, \eqref{eq: equiavlent EES <M} is an upper bound on the sum of the $k$-largest costs, i.e., EES. 
\section{Case Study: Richmond Water Network}
The Richmond water network is part of the Yorkshire water supply area in U.K. \cite{exeter}. More details about the model of this network is given in \cite{wang2021minimizing,ArastouCDC2024}. The control objective is given in \eqref{eq: SEMPC with EES}. The random vectors $\alpha_{l|t}$ are the electricity price, and it is assumed that i.i.d samples of them over the time horizon are available. The samples are drawn from a sum of a deterministic cost with two tariffs and a uniform random variable. A realization of electricity prices is shown in Fig. \ref{fig: electricity tariff over time}. 

The prediction horizon was $N=30$ hours, $\mathcal{X}=\{x\hspace{0.1 cm}:\hspace{0.1 cm}\mathunderbar{x}\leq x \leq \overline{x}\}$, $\mathcal{U}=\{u\hspace{0.1 cm}:\hspace{0.1 cm}\mathunderbar{u}\leq u \leq \overline{u}\}$, $\mathcal{X}^f=\{x \hspace{0.1 cm}:\hspace{0.1 cm} (x-x^s)^{\top} \Omega (x-x^s)\leq \kappa \}$ with $x^s=0.5(\mathunderbar{x}+\overline{x})$, $\Omega=I_{n}$, and $\kappa=0.8$. The water demand at time $t$ is expressed as $d_t = m_t \bar{d}$, where $m_t$ represents the demand multiplier illustrated in Fig. \ref{fig: The demand multiplier}. This multiplier has an average of 1, meaning that the average demand is $\bar{d}$, with different $\bar{d}$ for each individual demand.
\begin{figure}[t]
    \centering
    \begin{subfigure}[h]{0.22\textwidth}
        \centering        \includegraphics[width=\linewidth,,height=1.8cm]{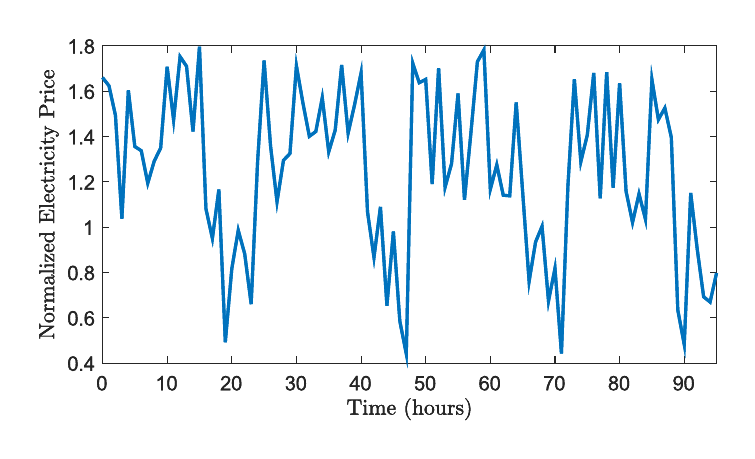}
        \caption{A realization of electricity prices}
        \label{fig: electricity tariff over time}
    \end{subfigure}
    \begin{subfigure}[h]{0.22\textwidth}
        \centering
        \includegraphics[width=\linewidth,,height=1.9 cm]{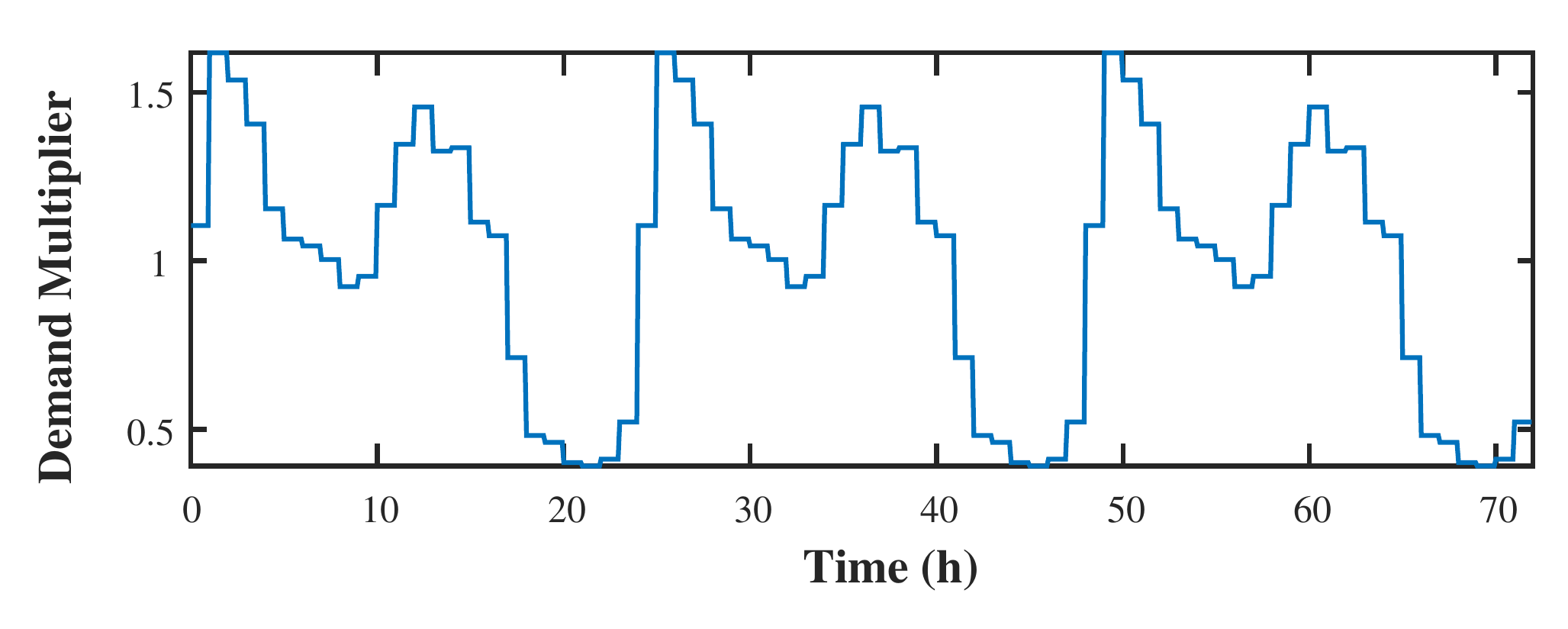}
        \caption{The demand multiplier}
        \label{fig: The demand multiplier}
    \end{subfigure}
    \caption{Electricity price and the demand multiplier}
    \label{fig:main}
\end{figure}


\subsection{Effect of Design Parameters on the Obtained Guarantees} \label{sec: sim guarantee for the whole region}

The probabilistic guarantee is obtained for the whole feasibility region and Corollary \ref{coro: Guarantee for solution of SEMPC} ensures that the guarantee can be used for the solution of \eqref{eq: Equivalent of SEMPC with EES}. To this end, the mapping in \eqref{eq: feasibility problem} is considered, and the number of support elements is obtained using Algorithm \ref{alg: support elements algorithm} and \eqref{eq: Finding support in P}. As the obtained number of support elements is a lower bound on the actual number; the guarantees are approximate. Note that the scenarios are changed in a moving horizon regime and the same procedure is carried out to find guarantees at each time step.


A comparison of the number of support elements within $\mathcal{P}$ characterized by \eqref{eq: state, input, and terminal sets} and the $\mathcal{U}^N$ box for different design parameters in Algorithm \ref{alg: support elements algorithm} is given in Table \ref{tbl: Comparison of suppurt elements with different inputs}. $k=2$, $\rho=\frac{\mu}{2}$, $N_r=3000$, and $\beta=10^{-6}$ in all of the cases. 


\begin{table}[h] 
\begin{center}    
\caption{Comparison of the number of support elements using various inputs in Algorithm \ref{alg: support elements algorithm}. $s_{\text{box}}$ and $s_{\mathcal{P}}$ are the number of support elements within the $\mathcal{U}^N$ box and $\mathcal{P}$.}
\label{tbl: Comparison of suppurt elements with different inputs}
\begin{tabular}{ccccccccc}
\hline
$N_s$ & $\mu$ & $ \overline{\beta}$ &$N_T$ &$s_{\text{box}}$ &$s_{\mathcal{P}}$ & $\mathunderbar{\epsilon}(s_{\mathcal{P}})$ &$\bar{\epsilon}(s_{\mathcal{P}})$ \\
\hline
2000  & $10^{-4}$ & $10^{-6}$ &733984 &45 &31 &0.005 &0.037 \\
2000  & $10^{-3}$ & $10^{-5}$ &57886 &31 &23 &0.003 &0.031 \\
10000  & $10^{-3}$ & $10^{-5}$ &57886 &39 &32 &0.001 &0.007 \\
\hline
\end{tabular}
\end{center}
\end{table}
The number of input samples in Table \ref{tbl: Comparison of suppurt elements with different inputs}, $N_T$, is computed from \eqref{eq: number of testing samples from input}. $N_T$ increased considerably, due to the reduced $\mu$ to provide a better guarantee for the obtained number of support elements as stated in Lemma \ref{lem: number of testing samples}. Since the testing stage in Algorithm \ref{alg: support elements algorithm} is easy to be carried out, increasing $N_T$ will not affect practicality of the algorithm significantly. 

\subsection{SEMPC with EES constraint}

The proposed control problem in \eqref{eq: Equivalent of SEMPC with EES} was applied to the Richmond water network with $M=2150$ and $N_s=2000$. The water level in tank A (a state) and the flow through pump A (an input) are shown in Fig. \ref{fig: x_A and u_A}. To show the effectiveness of the proposed method, the SEMPC cost and the EES are shown in Fig. \ref{fig:Avg and EES with and without EES cons} with and without the EES constraint.

\begin{figure}[t]
    \centering
    \includegraphics[scale=0.37]{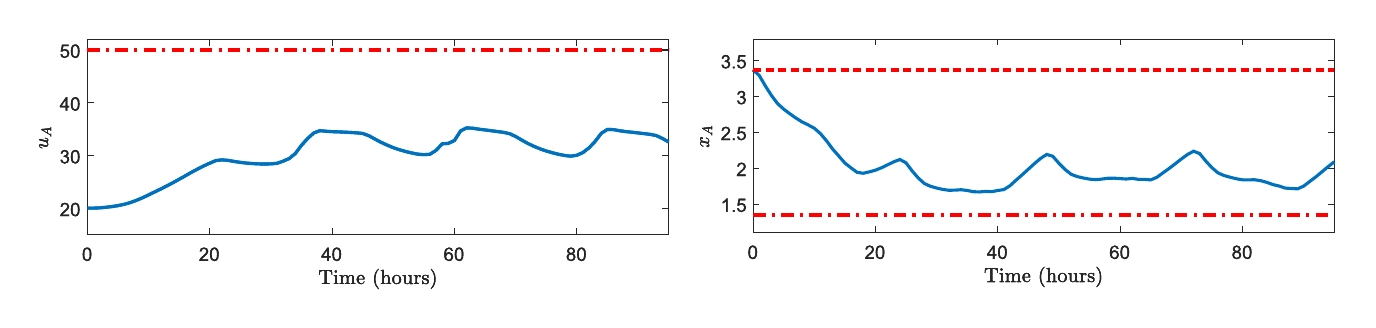}
    \caption{The simulation results for $u_A$ and $x_A$ (dashed red lines indicate constraints)}
    \label{fig: x_A and u_A}
\end{figure}

It can be seen from these figures that the EES constraint ensures that risk is below the desired value at each time step; however, it resulted in a higher SEMPC cost. The EES constraint resulted in high values for the other terms in the SEMPC cost, such as input smoothness term. Hence, some peaks with large values are visible in Fig. \ref{fig:Avg and EES with and without EES cons}. 
\begin{figure}
    \centering
    \includegraphics[scale=0.5]{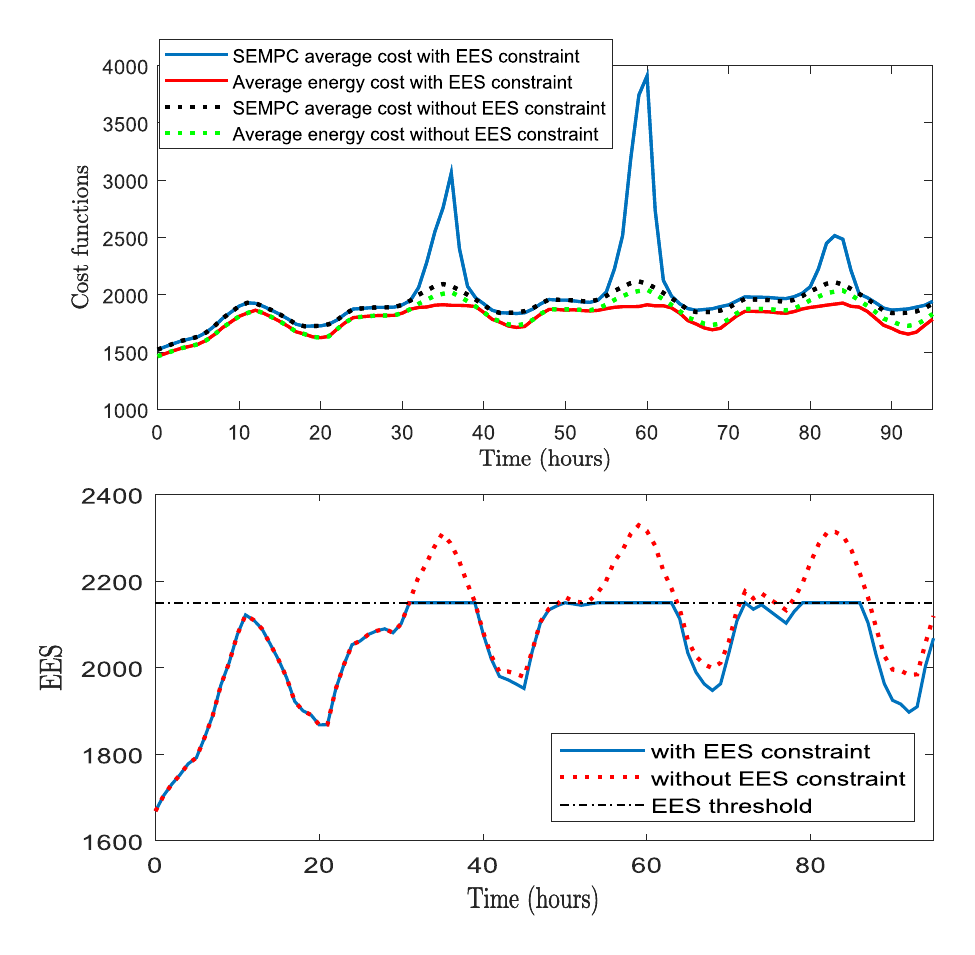}
    \caption{Average cost and EES with and without the EES constraint}
    \label{fig:Avg and EES with and without EES cons}
\end{figure}



Using the number of support scenarios found by Algorithm \ref{alg: support elements algorithm}, $\bar{\epsilon}(s_{\mathcal{P}})$ at the SEMPC solution was between 0.025 and 0.045 at each time where the values $k=2$, $\rho=\frac{\mu}{2}$, $N_r=3000$, $\mu=10^{-3}$, $\bar{\beta}=10^{-5}$, and $\beta=10^{-6}$ were used. The guarantee at each time step was obtained from Corollary \ref{coro: Guarantee for solution of SEMPC}.   

\section{Conclusion}
The paper introduces an SEMPC strategy that integrates an EES constraint to effectively manage risk while minimizing average cost. A probabilistic certificate is obtained for the solution of the SEMPC strategy using the number of support elements for the whole feasibility region of the control problem. The support elements are the scenarios that are among the largest $k$ costs for some input. Evaluating the cost function and ordering them for all combinations of inputs to find the number of support elements for a high dimensional system is cumbersome. Thus, a heuristic algorithm is proposed to address this challenge effectively. Moreover, the EES constraint increases the computational complexity of solving the control problem considerably. A reformulation is proposed in this paper to reduce the computational complexity. 




\bibliographystyle{ieeetr}

\bibliography{References}

\end{document}